\documentclass[conference]{IEEEtran}
\usepackage{graphicx} 
\usepackage[letterpaper, 
            left=0.65in,   
            right=0.65in,  
            top=0.75in, 
            bottom=1.05in] 
            {geometry}
\usepackage{amsmath,amsthm,amssymb}

\long\def\/*#1*/{}

\usepackage[normalem]{ulem}
\newcommand\redout{\bgroup\markoverwith{\textcolor{red}{\rule[.5ex]{2pt}{0.4pt}}}\ULon}

\newtheorem{theorem}{Theorem}
\newtheorem{proposition}{Proposition}

\newtheorem{lemma}{Lemma}

\newcommand{\Geomv}{\mathsf{Geom}}

\usepackage{texdef2020}
\usepackage{float}
\usepackage{xcolor}
\usepackage{tikz}
\usepackage{color,soul}
\usetikzlibrary{patterns}
\usepackage{algorithm}
\usepackage{algpseudocode}
\usepackage{hyperref}
\theoremstyle{remark}
\newtheorem{remark}{Remark}
\newcommand{\xtic}[1]{(#1,0) -- ++(0,-0.25)}

\newcommand{\I}[1]{I_{\set{#1}}}
\usepackage{cite}

\title{Age and Stability Trade-offs in Remote Monitoring Systems}
\author{%
 \IEEEauthorblockN{Nitya Sathyavageeswaran, Anand D.~Sarwate, Roy D.~Yates, Narayan Mandayam}
 \IEEEauthorblockA{Rutgers, The State University of New Jersey \\
                   Email: \{nitya.s, anand.sarwate\}@rutgers.edu, \{ryates, narayan\}@winlab.rutgers.edu}
}

\date{March 2026}

\begin{document}

\maketitle

\begin{abstract}
    Timely information is important in a wide variety of Internet of Things (IoT) services in which a shared server must manage two competing tasks: (i) processing a queue of jobs, and (ii) generating status updates to a remote monitor. This creates a fundamental trade-off between queue stability and data freshness. In this work, we model this scheduling decision as a Markov Decision Process (MDP) with the objective of minimizing a weighted sum of the average Age of Information (AoI) and the average queue length. We show that the optimal scheduling strategy is a queue-dependent age threshold which is monotonic. The shape of the switching curve differs according to different priority regimes. Finally, we compare the optimal MDP policy against heuristic policies.  
\end{abstract}

\section{Introduction}
In modern Internet of Things (IoT) systems, the utility of collected data degrades rapidly over time, making timely information delivery critical~\cite{shukla2023improving}.
As an example, consider an industrial application in which an installed device has a camera that is used to perform continuous high-speed quality control on an assembly line. Image processing is done locally and, to obtain high throughput, the local processor operates in a heavy-traffic regime.
However, the camera must also periodically transmit status updates to a monitor which is tracking the entire production line. Since these updates share the local processor, generating them temporarily degrades the primary image-processing job pipeline to a slower service. The severity of this delay depends on the type of update being sent; transmitting a simple summary update may be fast, but processing a high resolution frame or doing extensive post processing may be slow. 

In heavy traffic, even this brief interrupt for an update risks queue instability. This creates a fundamental trade-off between system stability and data freshness at the monitor. In this paper we use the Age of Information (AoI)~\cite{kaul2012real} metric to quantify the freshness of information at the monitor.  An \textit{always-update} policy prioritizes AoI but leads to queue instability, whereas an \textit{always-serve} policy stabilizes the queue but allows AoI to grow unbounded. Balancing these extremes requires a state-aware scheduling approach. 

In this paper we formulate the problem of balancing AoI and stability as a Markov Decision Process (MDP) in which a local shared processor must dynamically choose whether to service a queue holding an incoming stream of data processing jobs or generate a fresh status update for the monitor. We use a weighted sum of the average age and the expected queue length as a cost function. We show that the  optimal scheduling policy follows a state-aware monotonic switching curve that adapts to different parameter regimes. We empirically compare the optimal policy to two simple heuristics: a Memoryless policy that randomly chooses to send an update, and a Myopic policy that minimizes the expected cost in the next time step. Because our optimal policy does not have a closed form expression, we identify a low-complexity heuristic that approximates the switching curve and demonstrates that it can achieve near-optimal performance. 

\noindent \textbf{Related work.} Average age has been extensively studied for various queuing systems~\cite{kaul2012real, kam2018age,KamPathdiversity, yates2018age}. Prior work has examined scheduling strategies for minimizing long run average age in multi-user settings~\cite{kadota2018optimizing, sun2018age, Hsu2020Scheduling}.
MDP-based formulations have also been used to derive optimal age-aware policies~\cite{hsu2018age, CeranHARQ, tang2020minimizing, bedewy2021optimal, sathyavageeswaran2024timely}.  However, these works typically focus solely on freshness, without explicitly accounting for queue stability constraints under shared processing resources. Recent works have also begun to incorporate service delays and system dynamics into the design of update policies. In particular, sampling and scheduling problems with stochastic service times and queuing effects have been studied~\cite{ornee2019sampling, sun2017remote, sun2017update}.

\section{System Model and Problem Formulation}

\noindent \textbf{Notation.} Random variables will be denoted by capital letters and realizations by lower case letters. Time passes in integer slots with slot $n\ge0$ denoting the time interval $[n,n+1)$.

We study a discrete time system consisting of a data source, a status update generator, a shared local processor, a queue, and a remote monitor as shown in Fig.~\ref{fig:sysmodel}. 
\begin{figure}
    \centering
    \includegraphics[width=\linewidth]{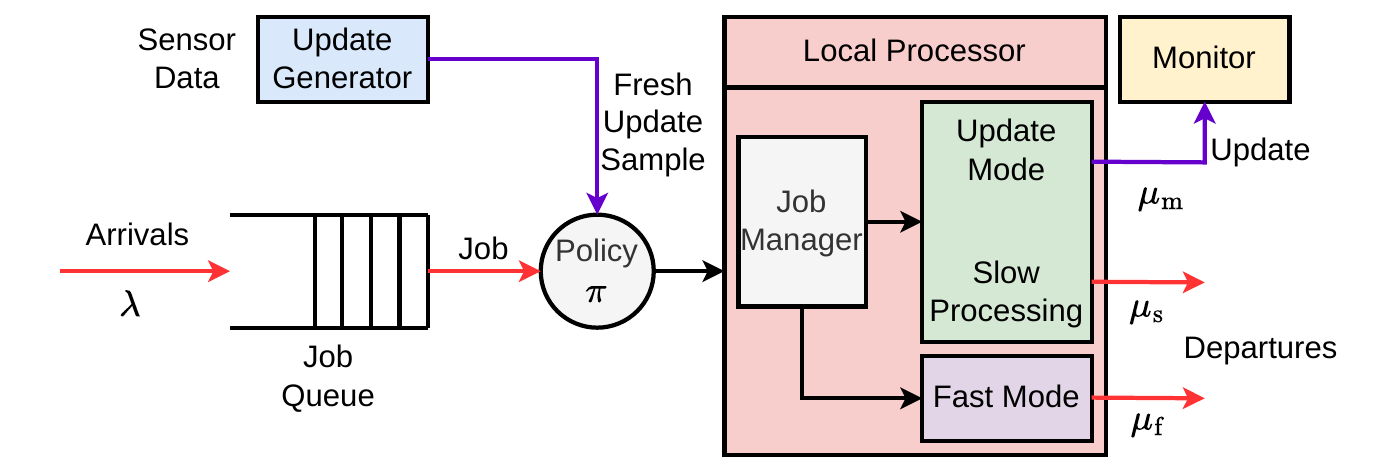}
    \caption{System model for the status update and the queueing process.}
    \label{fig:sysmodel}
\end{figure}Within a slot $n$, the following events happen in order:
\begin{enumerate}
    \item Arrivals: Data processing jobs arrive at a First-Come-First-Served (FCFS) queue. We model this as a Bernoulli process with rate $\lambda$, where $X_n = 1$ if an arrival occurs in slot $n$, and $X_n = 0$ otherwise. 
    \item Scheduling: The scheduler observes the system state (queue length $Q_n$ and age $A_n$), and chooses an action $U_n \in \{0,1\}$. 
    \item  Service and Delivery:
    \begin{itemize}
        \item If $U_n=0$ (\textit{Serve Queue}),  the local processor serves the queue {\em fast}. A job is successfully processed and departs with probability $\mu_\mathrm{f}$.
        \item If $U_n=1$ (\textit{Generate Update}), the local processor pulls a fresh update from the source and spends one time slot attempting to deliver to the monitor. This new status update is successfully delivered at the end of the slot with probability $\mu_\mathrm{m}$. Concurrently, the processor serves the data queue at a {\em slow} rate; a job departs with probability $\mu_\mathrm{s}\le \mu_\mathrm{f}$.
    \end{itemize} 
\end{enumerate}

\subsection{Problem Formulation}

The scheduler determines the server's operating mode in each slot. We model the scheduler's actions over time as a sequence $\mathbf{U} = (U_1, U_2, \ldots)$, where
\begin{align}
    U_n&=\begin{cases}
        1 &\text{Generate Update (Update Mode)}\\
        0 &\text{Serve Queue (Fast Mode)}.
    \end{cases} \label{eq:actionscheduler}
\end{align}

To  use the AoI metric to measure timeliness, let $R_n = \max \{ k \le n\colon \text{update generated at } k \text{ delivered} \}$ be the generation time of the most recent update at the monitor. The age at the monitor at slot $n$ is $A_n = n - R_n$. The age increases by 1  in every slot where a successful update does not occur.
The average age is defined as\footnote{Actions occur slot by slot in discrete time but age is averaged continuously over the slot; when $A_n=a$, the average age over slot $n$ is $a+1/2$.}
\begin{align}
    \Delta = \limty{N} \frac{1}{N}\sum_{n=1}^{N} (A_n +1/2) \eqnlabel{sum-age}.
    \end{align}

For the data processing job queue, an arrival and departure can occur in the same slot. The evolution of the queue backlog $Q_n$  is governed by arrivals $X_n$ and departures $D_n$:
\begin{align}
    Q_{n+1} = Q_n - D_n + X_n,
\end{align}
where the departure probability is given by 
\begin{IEEEeqnarray}{rCl}
    P(D_n=1 \mid U_n, Q_n, X_n) = \begin{cases}
        \mu(U_n) & \text{if } Q_n + X_n > 0 \\
        0 &\text{if } Q_n + X_n = 0,
    \end{cases}\IEEEeqnarraynumspace
\end{IEEEeqnarray}
where $\mu(0) = \mu_\mathrm{f}$  and $\mu(1) = \mu_\mathrm{s}$. Similarly, we define the average queue length as:
\begin{align}
    \bar{Q} = \lim_{N \to \infty} \frac{1}{N}\sum_{n=1}^{N} Q_n. \eqnlabel{avg-queue}
\end{align}

Our goal is to minimize the weighted sum of the average age and the average queue length. The average cost is defined as:
\begin{align}
    \bar{C} = \lim_{N \to \infty} \frac{1}{N}\sum_{n=1}^{N} (A_n + \frac{1}{2} + w Q_n),
\end{align}
where $w > 0$ is a weight parameter representing the penalty for queue backlog relative to staleness.

Given these ingredients, we formulate the problem as a Markov Decision Process (MDP). We define the state of the system as $S_n = (Q_n, A_n)$. The MDP consists of a tuple $(\mathcal{S}, \mathcal{U}, P, C)$: $s=(q, a) \in \mathcal{S}$, where $q \in \mathbb{N}_0$ and $a \in \mathbb{N}$ are the countably infinite set of states, and $u \in \mathcal{U}$ is the action defined in \eqref{eq:actionscheduler}.

We assume job queue arrivals are a rate $\lambda$ Bernoulli process, independent of data service and update generation. The transition function $P((q', a') \mid (q, a), u)$ governs the probability of the transition $S_n \to S_{n+1}$. Since only one arrival or departure can happen per slot, the transition probabilities $P(q \to q'\mid u)$ conditioned on the action $U_n=u$ are given as follows
\begin{align}
P(q\rightarrow q'\mid u) = 
\begin{cases} 
\lambda(1 - \mu(u)) & q' = q+1 \\
(1 - \lambda)\mu(u) \I{q > 0} & q' = q-1 \\
\begin{aligned}
&1 - \lambda(1 - \mu(u)) \\
&- (1 - \lambda)\mu(u) \I{q > 0}
\end{aligned} & q' = q \\
0 & \text{otherwise}.
\end{cases}
\end{align}

Based on the chosen action, the joint state transitions are:
\begin{itemize}
    \item If $u = 0$ (Serve Queue), the age increases by 1: 
    \begin{align}
        P((q', a+1) \mid (q, a), 0) = P(q \to q' \mid 0).
    \end{align}
    \item If $u = 1$ (Generate Update), the age resets to 1 with probability $\mu_\mathrm{m}$:
    \begin{align}
        P((q', 1) \mid (q, a), 1) &= \mu_\mathrm{m} P(q \to q' \mid 1), \\
        P((q', a+1) \mid (q, a), 1) &= (1 - \mu_\mathrm{m}) P(q \to q' \mid 1).
    \end{align}
\end{itemize}

We model the immediate cost $C$ of taking action $U$ in state $S$ as
\begin{align}
    C_n= C(S_n, U_n) = A_n + \frac{1}{2} + w Q_n.
\end{align}
We want to find a policy $f^*$ that minimizes the long term average cost given by
\begin{align}
    V_{f}(s) = \limsup_{N\to \infty}\frac{1}{N}\Eop_{f}\left[ \sum_{n=1}^{N}  C(S_n,U_n) \middle\vert S_0 = s \right]. \eqnlabel{eq:avgcost}
\end{align}

\section{Heuristic Policies}

We first analyze a few heuristic policies to understand the age/queue length tradeoff. 

\subsection{Memoryless inspection}
\label{sec:memoryless}
In this policy, the scheduler makes its decisions independently of the system state. In every slot, the scheduler chooses to generate a status update with fixed probability $p$. 
The average age at the monitor for this policy is given by 
\begin{align}
    \Delta_\mathrm{mem} & = \frac{1}{p\mu_\mathrm{m} } +  \frac{1}{2}\eqnlabel{agememoryless}. 
\end{align}
The detailed derivation of this average age is provided in appendix~\ref{app:memoryless}. 
The effective service rate, $\mu _\mathrm{mem}$ for this policy, which is the probability of a job departure from the queue in a slot is given by
\begin{align}
    \mu_{\mathrm{mem}} &= p \mu_{\mathrm{s}} + (1-p) \mu_{\mathrm{f}} \eqnlabel{mueffmemoryless}. 
\end{align}

The queue evolves as a discrete-time Geo/Geo/1 queue. Using standard results for discrete-time Markov queues~\cite{ross2014introduction}, the steady-state average queue length is  given by 
\begin{align}
    \bar{Q}_\mathrm{mem} = \frac{\lambda(1-\mu_{\mathrm{mem}})}{\mu_{\mathrm{mem}} - \lambda} \eqnlabel{qlengthmem}.
\end{align}
A fundamental requirement for the system is to maintain queue stability, ensuring that the average queue length $\bar{Q}$ remains finite. Given the arrival rate $\lambda$ and the effective service rate $\mu_{\mathrm{mem}}$, the system is stable if and only if $\lambda < \mu_\mathrm{mem}$. This implies that the update probability must satisfy
\begin{align}
 p<   p_{\max}\triangleq  \frac{\mu_{\mathrm{f}} - \lambda}{\mu_{\mathrm{f}} - \mu_{\mathrm{s}}} \eqnlabel{stability_memoryless}.
\end{align}

From \eqnref{agememoryless}, we see that the average age decreases as the update probability $p$ increases. However, a higher $p$ reduces the effective service rate which in turn increases the queue backlog. This competing dynamic implies there is a unique update probability $p^*$ that optimally resolves this tradeoff, which is given in the following proposition.

\begin{proposition}
\label{lem:optimal_p}
    For the Memoryless policy where updates are chosen with probability $p$ in each slot, the probability $p^*$ that minimizes the average cost $\bar{C}(p)$ is
\begin{align}
    p^* = \min\left(1, \frac{\mu_{\mathrm{f}} - \lambda}{\Delta\mu + \sqrt{w\lambda(1-\lambda)\mu_{\mathrm{m}}\Delta\mu}} \right),
\end{align}
where $\Delta\mu = \mu_{\mathrm{f}} - \mu_{\mathrm{s}}$, $\lambda < \mu_{\mathrm{f}}$. If $\lambda \ge \mu_{\mathrm{f}}$, no stable policy exists.
\end{proposition}

The proof is provided in Appendix~\ref{app:lemma_optimal_p}.

\subsection{Myopic policy}
In this policy the scheduler minimizes the expected cost in the next time slot. At any time $n$ with state $S_n = (q,a)$, the scheduler selects the action $u\in\{0,1\}$ that minimizes the one-step expected cost $\E{C_{n+1} \mid S_n = (q,a), U_n = u}$.

By evaluating the expected next-step costs for updating versus serving the queue (detailed in Appendix~\ref{app:myopic_derivation}), we find that the Myopic scheduler behaves as a deterministic threshold policy. It generates a status update ($U_n = 1$) if the current age $A_n= a$ satisfies
\begin{align}
    a > \begin{cases} 
      \frac{w\lambda(\mu_{\mathrm{f}} - \mu_{\mathrm{s}})}{\mu_{\mathrm{m}}} & \text{if } q = 0 \\
      \frac{w(\mu_{\mathrm{f}} - \mu_{\mathrm{s}})}{\mu_{\mathrm{m}}} & \text{if } q > 0.
   \end{cases}
\end{align}
Otherwise, the scheduler chooses to serve the queue ($U_n=0$). We denote the deterministic age threshold for a back-logged queue ($q>0$) as $a^*_\mathrm{myopic} = w(\mu_\mathrm{f} - \mu_\mathrm{s})/ \mu_\mathrm{m}$.

To establish the stability of the Myopic policy, a standard approach would rely on Foster's stability criterion~\cite{foster1953stochastic} using a Lyapunov function such as $L(Q_n, A_n) = Q_n$. Basic versions of this criterion require the expected one-step drift to be strictly negative outside a finite set of states. However, under the Myopic Policy when $A_n > K$, the expected one-step queue drift, $\lambda - \mu_\mathrm{s}$ is strictly greater than zero.
Also the age process lacks autonomy since the threshold changes when $Q_n = 0$. To resolve this we adopt a methodology inspired by Yates et~al.~\cite{yates2015energy} and construct a fictitious autonomous age process. This lets us establish queue ergodicity using the Markov-modulated framework of Foss et~al.~\cite{foss2013stability}.

\begin{theorem}[Myopic Stability region]
\label{thm1}
    Under the Myopic scheduling policy, the queue is stable (positive recurrent) if the arrival rate satisfies the following condition:
    \begin{align}
        \lambda < \frac{K \mu_\mathrm{f} \mu_\mathrm{m} + \mu_\mathrm{s}}{K\mu_\mathrm{m} + 1},
    \end{align}
    where $K = \lfloor a^*_{\mathrm{myopic}} \rfloor$ is the deterministic age threshold for a backlogged queue.
\end{theorem}

We construct a fictitious system applying a constant threshold $K$, letting age operate as an independent renewal cycle. By renewal reward theory, its effective service rate is 
\begin{align}
     \mu_{\mathrm{eff}} = \frac{K \mu_{\mathrm{f}} + \mu_{\mathrm{s}}/\mu_{\mathrm{m}}}{K + 1/\mu_{\mathrm{m}}} = \frac{K \mu_{\mathrm{f}} \mu_{\mathrm{m}} + \mu_{\mathrm{s}}}{K\mu_{\mathrm{m}} + 1}.
\end{align}
By Corollary~1 of Foss et~al.~\cite{foss2013stability}, this independent system is stable if $\lambda < \mu_{\mathrm{eff}}$. Since the true Myopic policy clears backlogged jobs at least as fast as the fictitious system, it is inherently stable (full proof in Appendix~\ref{App:myp_stability}).

\section{Optimal Policy Analysis}

We now turn to finding the optimal policy $f^*$ which minimizes the long term average cost in \eqnref{eq:avgcost}. We first define the discounted cost $V_{f,\beta}(s)$ as follows:
\begin{align}
    V_{f,\beta}(s) = \Eop_{f} \left[ \sum_{t=0}^\infty \beta^t C(S_t, U_t) \middle\vert S_0 = s \right], \eqnlabel{eq:discounted_cost}
\end{align}
where $s=(q,a)$ is the initial system state and $\beta \in (0,1)$ is the discount factor. We define the optimal discounted value function
\begin{align}
    V_\beta(s) = \inf_{f} V_{f,\beta}(s).
\end{align}

If the value function $V_{\beta}(s)$ is finite, it satisfies the following discounted cost optimality equation:
\begin{align}
    V_\beta(s) &= \min_{u \in \{0,1\}} \Bigl( C(s,u) + \beta \sum_{s'} P_{ss'}(u) V_\beta(s') \Bigr), \label{eq:bellman}
\end{align}
for all $s \in \mathcal{S}$, where $P_{ss'}(u)$ is the transition probability from state $s$ to $s'$ under action $u$. 

We can define a value iteration $V_{\beta,n}(s)$ recursively as
\begin{align}
    &V_{\beta,n}(q,a) \nonumber \\
    &\quad= \min_{u \in \{0,1\}} \Bigl( C((q,a),u) + \beta \sum_{s'} P_{ss'}(u) V_{\beta,n-1}(s') \Bigr) \nonumber \\
    &\quad= a + wq + \min \Bigl( \beta \E{V_{\beta,n-1}(Q'_0, a+1)}, \nonumber \\
    &\qquad \qquad \beta \mu_{\mathrm{m}} \E{V_{\beta,n-1}(Q'_1, 1)} \nonumber \\
    &\qquad \qquad + \beta (1-\mu_{\mathrm{m}}) \E{V_{\beta,n-1}(Q'_1, a+1)} \Bigr), \label{eq:valuefunction}
\end{align}
for any $n \ge 1$, with $V_{\beta,0}(s)=0$ for all $s \in \mathcal{S}$ and $\beta \in (0,1)$, where $Q'_0$ and $Q'_1$ denote the next-slot queue lengths under Fast Mode and Update Mode, respectively.

Under Sennott's conditions~\cite{sennott1989average} (verified in Appendix~\ref{app:sennott}), there exists a differential cost function $h(s)$ and a stationary deterministic policy $f^*$ satisfying the average cost optimality equation
\begin{align}
 g + h(s) = \min_{u \in \{0,1\}} \Bigl( C(s,u) + \sum_{s'} P_{ss'}(u) h(s') \Bigr),
\end{align}
for all $s \in \mathcal{S}$, where $g$ is the minimum average cost.

\begin{remark}[Stability of the Optimal Policy]
Since the immediate cost function explicitly penalizes the queue length ($C(S_n, U_n) = A_n + 1/2 + w Q_n$ with $w > 0$), any policy that yields a finite average cost must inherently maintain a finite average queue length. In Section~\ref{sec:memoryless} we see that whenever $\lambda < \mu_{\mathrm{f}}$, a Memoryless policy with update probability $p < p_{\max}$ stabilizes the queue and achieves a finite average cost. By definition, the optimal policy $f^*$ achieves an average cost no greater than this Memoryless baseline. Therefore, the optimal policy guarantees a finite average queue length, and ensures strict system stability for any $\lambda < \mu_{\mathrm{f}}$.
\end{remark}

\subsection{Relative Value Iteration (RVI)}

 To numerically evaluate the optimal policy $f^*$, we truncate the state space to a finite grid $\mathcal{S}_{\mathrm{trunc}} = \{0, \dots, Q_{\max}\} \times \{1, \dots, A_{\max}\}$. The optimal policy for the truncated MDP is identical to that of the original infinite MDP as $Q_{\max}, A_{\max} \to \infty$, provided that the truncated MDP is \textit{unichain}~\cite{puterman1994markov}. Our truncated MDP satisfies this unichain condition because, under any stationary policy, the Markov chain consists of a single recurrent class. The state $(0,1)$ can be reached from all other states $(q,a) \in \mathcal{S}_{\mathrm{trunc}}$. This is because, from any state, there is a strictly positive probability of having a finite sequence of zero arrivals coupled with consecutive fast service completions, followed by a successful update.

We get the exact optimal policy over $\mathcal{S}_{\mathrm{trunc}}$ by applying RVI~\cite{puterman1994markov}. The value function is initially set to $V_0(s) = 0$ for all states $s \in \mathcal{S}_{\mathrm{trunc}}$. At each iteration $n \geq 0$, we update:
\begin{align}
    V_{n+1}(s) &= \min_{u} \left[ C(s,u) + \E{V_n(s')} \right] - V_n(s_{\mathrm{ref}}),
\end{align}
where $s_{\mathrm{ref}} = (0,1)$ is a fixed reference state. The iteration stops when the condition $\max_{s \in \mathcal{S}_{\mathrm{trunc}}} |V_{n+1}(s) - V_n(s)| < \epsilon$ is satisfied.

Fig.~\ref{fig:policy_map} and Fig.~\ref{fig:policy_map0.008} illustrate the exact optimal scheduling policy derived via RVI over the state space for three different values of the queue penalty weight $w \in \{0.1, 1, 10\}$. The optimal policy is a non-decreasing switching curve $a^*(q)$ which partitions the state space into two distinct operating regions. The blue region below the boundary indicates states where the optimal action is to serve the queue in Fast Mode ($u=0$) to clear backlog, and the red region above the boundary indicates states where the optimal action is to generate an update ($u=1$) to prioritize data freshness. As the penalty for queue backlog $w$ increases, this monotonic boundary (age threshold) shifts upward, which indicates the scheduler's hesitance to interrupt queue service for status updates.

Note: The dashed lines in Fig.~\ref{fig:policy_map} and Fig.~\ref{fig:policy_map0.008} represent a low-complexity structural heuristic that approximates this boundary, which we introduce and analyze in Section~\ref{sec:heuristic}.

\begin{figure}
    \centering
    \includegraphics[width=\linewidth]{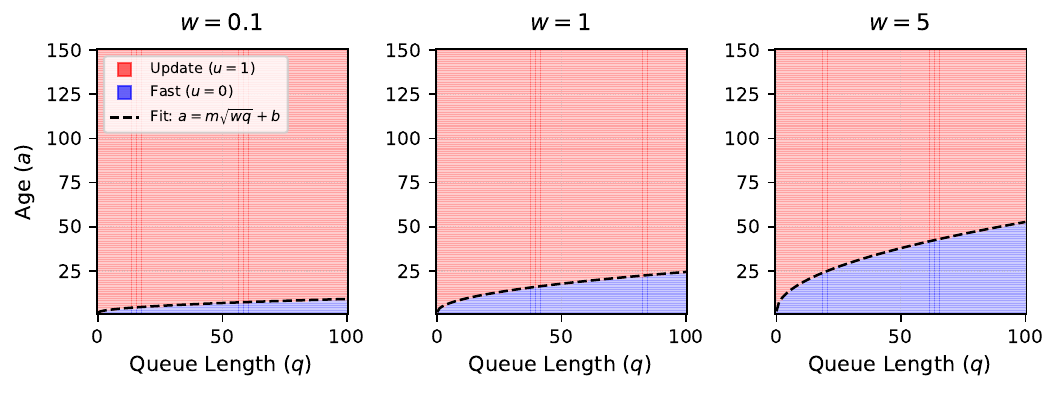}
    \caption{The optimal switching boundaries $a^*(q)$ for varying queue penalty weights $w \in \{0.1, 1, 5\}$. For any given weight curve, the region above the boundary indicates the Update Mode ($u=1$), while the region below indicates the Fast Mode ($u=0$). The system parameters are fixed at $\lambda = 0.4$, $\mu_{\mathrm{f}} = 0.5$,  $\mu_{\mathrm{s}} = 0.3$, and $\mu_{\mathrm{m}} = 0.8$.}
    \label{fig:policy_map}
\end{figure}

\begin{figure}
    \centering
    \includegraphics[width=\linewidth]{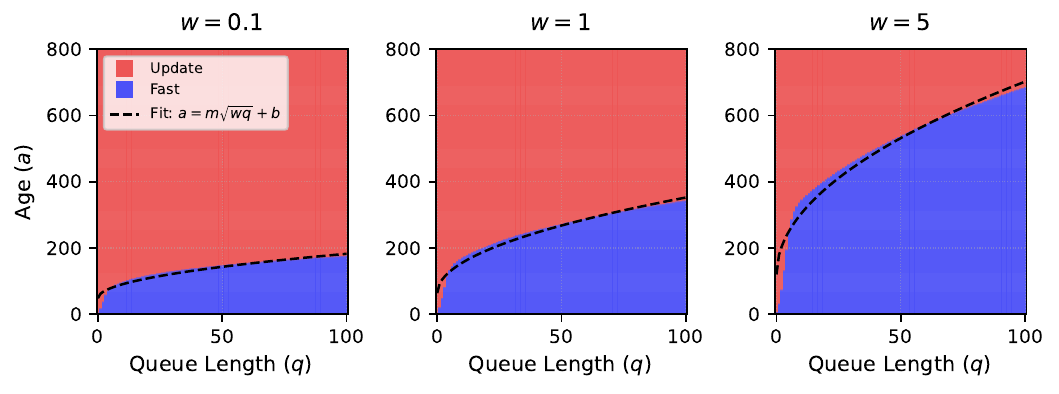}
    \caption{The optimal switching boundaries $a^*(q)$ for varying queue penalty weights $w \in \{0.1, 1, 5\}$. For any given weight curve, the region above the boundary indicates the Update Mode ($u=1$), while the region below indicates the Fast Mode ($u=0$). The system parameters are fixed at $\lambda = 0.4$, $\mu_{\mathrm{f}} = 0.5$,  $\mu_{\mathrm{s}} = 0.3$, and $\mu_{\mathrm{m}} = 0.008$.}
    \label{fig:policy_map0.008}
\end{figure}

\subsection{Low-Complexity Square-Root Heuristic}
\label{sec:heuristic}

While the exact MDP formulation yields an optimal policy $f^*$, storing an infinite-sized lookup table is impractical for implementation on resource-constrained edge devices. Our empirical evaluations using RVI reveal that the optimal switching boundary exhibits a distinct square-root shape. To reduce the computational complexity, we propose a low-complexity heuristic based on a square-root curve.

We define the Square-Root Heuristic, denoted as $f_{\mathrm{sqrt}}$, which maps the current system state $S_n = (q, a)$ to an action $u \in \{0, 1\}$ as follows:
\begin{align}
    f_{\mathrm{sqrt}}(q, a) = 
    \begin{cases} 
        1 & \text{if } a > m\sqrt{wq} + b \\
        0 & \text{otherwise},
    \end{cases}
    \eqnlabel{eq:sqrt_heuristic}
\end{align}
where $m > 0$ and $b \ge 0$ are tunable parameters. 

Here, $b$ represents the baseline age tolerance of the system when the queue is empty ($q=0$). The parameter $m$ governs the scheduler's sensitivity to queue backlog. A larger $m$ dictates a steeper initial penalty for the queued jobs, delaying status updates to prioritize queue stability. As shown by the dashed lines in Fig.~\ref{fig:policy_map}, and Fig.~\ref{fig:policy_map0.008}, this heuristic uses just two parameters to approximate the optimal policy, avoiding the need for the lookup table over the infinite space. The optimal parameters $(m^*, b^*)$ can be determined offline by applying ordinary least squares regression to the exact MDP switching boundary. Alternatively, for systems where computing the full MDP is computationally prohibitive, these parameters can be found directly via a low-dimensional grid search over the target stability region.

We compare optimal MDP policy and our square-root heuristic against two simpler heuristic baselines: the Memoryless policy and the Myopic policy. 
Fig.~\ref{fig:baselines} and Fig.~\ref{fig:baselinesmum0.008} illustrate the trade-off between the average queue length $\bar{Q}$ and the average Age of Information $\Delta$ at the monitor under fast ($\mu_{\mathrm{m}} = 0.8$) and severely constrained ($\mu_{\mathrm{m}} = 0.008$) update channels.
\begin{figure}
    \centering
    \includegraphics[width= 0.8\linewidth]{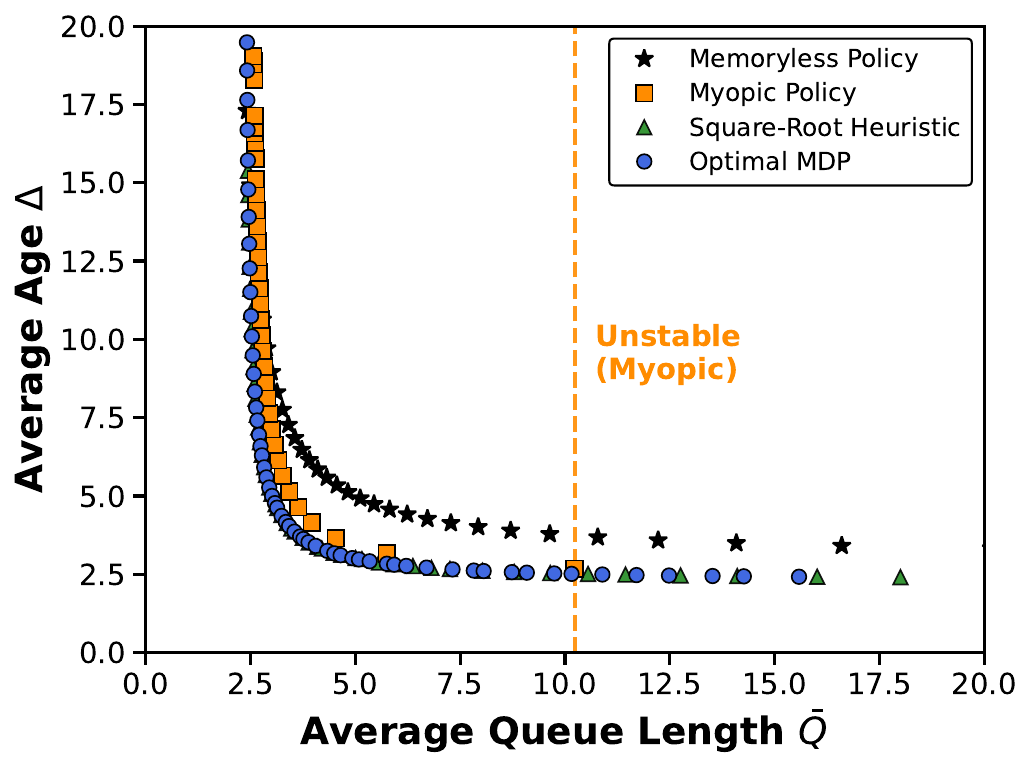}
    \caption{Performance comparison of the Optimal MDP, Square-Root heuristic, Myopic, and Memoryless policies. The Pareto frontiers are generated by varying the queue penalty weight $w$ for the Optimal MDP, the deterministic age threshold $K$ for the Myopic policy, and the update probability $p$ for the Memoryless policy. System parameters $\lambda = 0.4$, $\mu_{\mathrm{f}} = 0.5$, $\mu_{\mathrm{s}} = 0.3$, and $\mu_{\mathrm{m}} = 0.8$.  }
    \label{fig:baselines}
\end{figure}
\begin{figure}
    \centering
    \includegraphics[width= 0.8\linewidth]{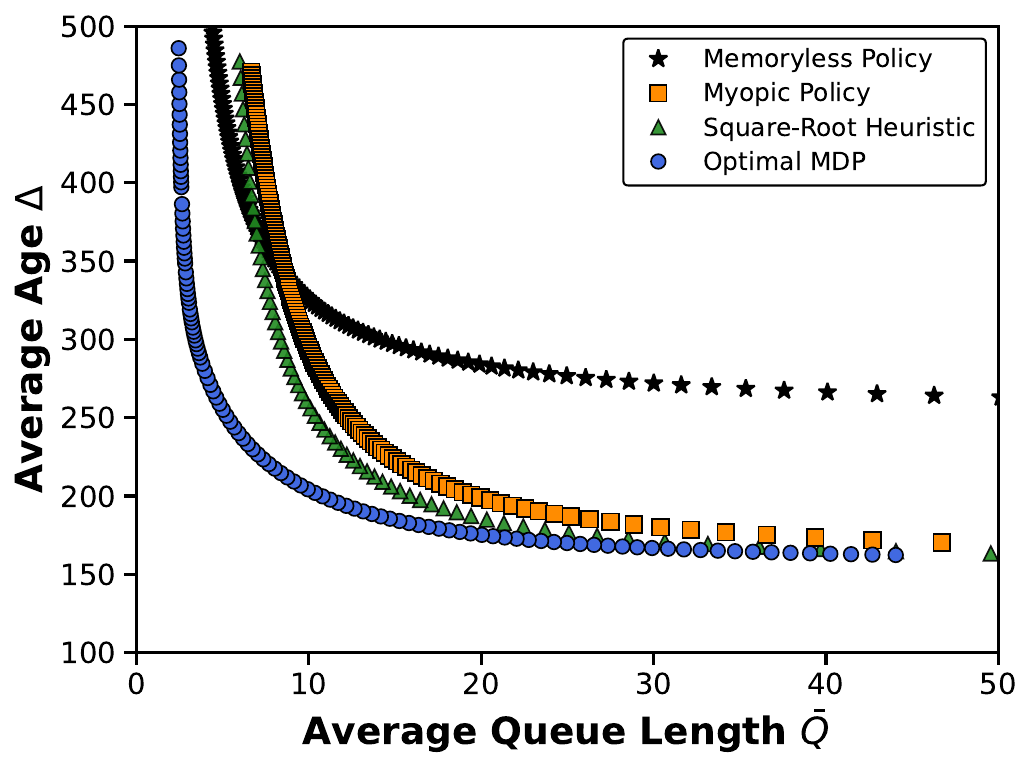}
    \caption{Performance comparison of the Optimal MDP, Square-Root heuristic, Myopic, and Memoryless policies under a severely constrained update channel. System parameters $\lambda = 0.4$, $\mu_{\mathrm{f}} = 0.5$, $\mu_{\mathrm{s}} = 0.3$, and $\mu_{\mathrm{m}} = 0.008$.}
    \label{fig:baselinesmum0.008}
\end{figure}
The optimal MDP curve is generated by using RVI across varying queue penalty weights $w$, establishing the best achievable Pareto frontier. As $w$ increases, the MDP heavily penalizes buffered jobs, minimizing $\bar{Q}$ while allowing $\Delta$ to grow. Conversely, as $w$ approaches zero, the policy prioritizes data freshness, minimizing $\Delta$ at the cost of a larger queue. Our simple Square-Root heuristic achieves near-optimal performance, overlapping with the Optimal MDP curve across the entire stable operating region. 

For the Memoryless policy, the average age and queue length are plotted using the exact analytical expressions derived in \eqnref{agememoryless} and \eqnref{qlengthmem}. The update probability $p$ is varied continuously from $0$ up to the stability limit $p_{\max}$. As $p$ increases, the system generates updates more frequently, and the average age reduces. However, because the policy makes scheduling decisions independently of the current queue backlog or data staleness, it yields the worst performance. 

For the Myopic policy, the average age and queue length are evaluated numerically by computing the stationary distribution of the discrete-time Markov chain induced by the deterministic age threshold $K$. As $w \to 0$, the threshold $K$ reduces, forcing the system to spend more time in the slow Update Mode ($\mu_{\mathrm{s}} = 0.3$). Since $\lambda = 0.4 > \mu_{\mathrm{s}}$, the effective service capacity falls below the stability bound from Theorem~\ref{thm1}. The vertical dashed line in Fig.~\ref{fig:baselines} marks this theoretical breaking point, beyond which the Myopic queue grows unbounded. 

The impact of the update rate $\mu_\mathrm{m}$ is stark when comparing the two regimes. Under fast updates (Fig.~\ref{fig:baselines}), the Square-Root heuristic achieves near-optimal performance, overlapping the MDP curve almost entirely. However, in the severely constrained environment (Fig.~\ref{fig:baselinesmum0.008}), an update takes an expected $1/\mu_\mathrm{m} = 125$ slots, and a performance gap emerges between the Square-Root heuristic and the Optimal MDP. As seen in Fig.~\ref{fig:policy_map0.008} for $w=5$, the optimal policy remains highly update-averse to prevent queue overflow adopting a boundary that is offset from the heuristic. Also, the heuristic becomes infeasible for queue lengths below a certain threshold.
This highlights the need of an exact, state-aware scheduling policy when managing extreme communication bottlenecks.

\section{Conclusion}

In this paper, we studied the fundamental trade-off between data freshness and queue stability in remote monitoring IoT systems. We formulate the scheduling problem as a Markov Decision Process, and show that the optimal policy for minimizing a weighted sum of Age of Information (AoI) and queue length  follows a state-aware, monotonic switching curve. Our numerical evaluations reveal that this optimal boundary is well-approximated by a low-complexity square-root heuristic, except when the update channel is severely constrained. Furthermore, we compared the optimal approach against Memoryless and Myopic baselines, illustrating that state-aware scheduling not only improves the Pareto frontier but also drastically expands the stable operating region of the system.

\section*{Acknowledgements}

This material is based upon work supported by the National Science Foundation under grant CNS-2148104 and is supported in part by funds from federal agency and industry partners as specified in the Resilient \& Intelligent NextG Systems (RINGS) program.
\clearpage



\IEEEtriggeratref{11}
\bibliographystyle{IEEEtran}
\bibliography{agebib.bib}

\newpage\clearpage
\section{Appendix}

\subsection{Memoryless Inspection Age Analysis}
\label{app:memoryless}

The sample age process at the monitor for this policy is shown in Figure~\ref{Memoryless}. 
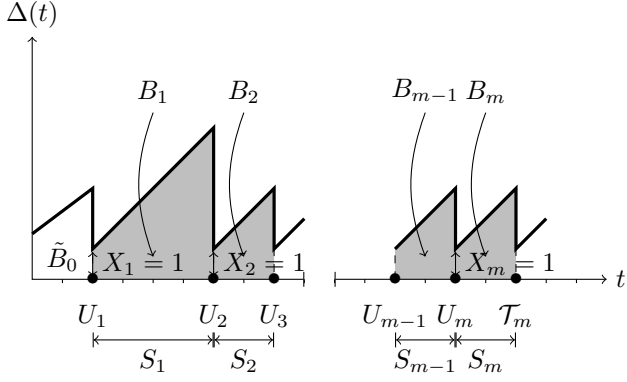
\begin{figure}
    \centering
    \begin{tikzpicture}[scale=0.2]
\draw [fill=lightgray, ultra thin, dashed] (2,0) to (2,2) to (10,10) to (10,0) ;
\draw [fill=lightgray, ultra thin, dashed] (10,0) to (10,2) to (14,6) to (14,0) ;
\draw [fill=lightgray, ultra thin, dashed] (22,0) to (22,2) to (26,6) to (26,0) ;
\draw [fill=lightgray, ultra thin, dashed] (26,0) to (26,2) to (30,6) to (30,0) ;

\draw [<-|] (-2,16) node [above] {$\Delta(t)$} -- (-2,0) -- (16,0);
\draw [|->] (18,0) -- (36,0) node [right] {$t$};

\draw [thin]\xtic{0}\xtic{2}\xtic{4}\xtic{6}
\xtic{8}\xtic{10}\xtic{12}\xtic{14}\xtic{16} \xtic{18}\xtic{20}\xtic{22}\xtic{24}
\xtic{26}\xtic{28}\xtic{30}\xtic{32}\xtic{34}; 

\draw [<->] (2,1.8) -- node [right] {$X_1=1$} (2,0.2);
\draw [<->] (10,1.8) -- node [right] {$X_2=1$} (10,0.2);
\draw [<->] (26,1.8) -- node [right] {$X_m=1$} (26,0.2);

\draw (0,1.5) node {$\Tilde{B}_0$};
\draw[<-] (6,1.5) to [out=110,in=250] (6,11) node [above] {$B_1$};
\draw[<-] (12,1.5) to [out=110,in=250] (12,11) node [above] {$B_2$};
\draw[<-] (24,1.5) to [out=110,in=250] (24,11) node [above]{$B_{m-1}$};
\draw[<-] (28,1.5) to [out=110,in=250] (28,11) node [above]{$B_{m}$};

\draw [very thick] (-2,3) -- (2,6) -- (2,2) -- (10,10) -- (10,2)  -- (14,6) -- (14,2) -- (16,4);
\draw [very thick] (22,2)  -- (26,6) -- (26,2) -- (30,6) -- (30,2) -- (32,4); 

\draw (2,0) node {$\bullet$} node [below=6pt] {$U_1$};
\draw (10,0) node {$\bullet$} node [below=6pt] {$U_2$};
\draw (14,0) node {$\bullet$} node [below=6pt] {$U_3$};
\draw (22,0) node {$\bullet$} node [below=6pt] {$U_{m-1}$};
\draw (26,0) node {$\bullet$} node [below=6pt] {$U_m$};
\draw (30,0) node {$\bullet$} node [below=6pt] {$\mathcal{T}_m$};

\draw  [|<->|] (2,-4) to node [below] {$S_1$} (10,-4);
\draw  [|<->|] (10,-4) to node [below] {$S_2$} (14,-4);
\draw  [|<->|] (22,-4) to node [below] {$S_{m-1}$} (26,-4);
\draw  [|<->|] (26,-4) to node [below] {$S_m$} (30,-4);
\end{tikzpicture}
    	\caption{Sample variation of the age process at the monitor for the Memoryless inspection policy.} 
\label{Memoryless}
\end{figure}
The average age at the monitor is the area under the saw tooth function in Figure~\ref{Memoryless} normalized by the total time of observation $\mathcal{T}_m$. The total area under the function is the sum of areas $\Tilde{B}_0, B_1, \ldots, B_m$.
\begin{align}
    \Delta_{\mathcal{T}_m} = \frac{\Tilde{B}_0 + \sum_{i=1}^m B_i}{\mathcal{T}_m} \eqnlabel{deltatm}
\end{align}

The area $B_i$ can be calculated as the sum of the area of $2$ different entities: i) a rectangle with height $X_i$ (where $X_i$ is the age of the newly delivered update at the time of delivery) and whose base connects the update delivery points $U_i$ and $U_{i+1}$, ii) an isosceles triangle representing the continuous age growth with width $S_i$ and whose base connects the points $U_i$ and $U_{i+1}$. We then get
\begin{align}
    B_i = X_i S_i + \frac{1}{2}S_i^2 \eqnlabel{biarea}
\end{align}

Substituting \eqnref{biarea} into~\eqnref{deltatm} and rearranging some terms yields
\begin{align}
    \Delta_{\mathcal{T}_m} = \frac{\Tilde{B}_0}{\mathcal{T}_m} + \frac{m}{\mathcal{T}_m}\frac{\sum_{i=1}^m \left(X_i S_i + \frac{1}{2}S_i^2\right)}{m}
\end{align}

As $\mathcal{T}_m\to \infty$, the first term vanishes and we get
\begin{align}
    \lim_{\mathcal{T}_m\to \infty} \frac{m}{\mathcal{T}_m} = \frac{1}{\E{S}}
\end{align}

The summation term is a sample average and it will converge to its stochastic average when $\mathcal{T}_m\to \infty$. Since the queue inspection process is independent of the update generation process, we have $\E{X_i S_i} = \E{X}\E{S}$. The average age at the monitor can now be calculated as follows:
\begin{align}
    \Delta &= \lim_{\mathcal{T}_m\to \infty} \Delta_{\mathcal{T}_m} \nonumber \\
    &= \frac{\E{X}\E{S} + \frac{1}{2}\E{S^2}}{\E{S}} \nonumber \\
    &= \E{X} + \frac{1}{2}\frac{\E{S^2}}{\E{S}} \label{eq:mem_age_expectation}
\end{align}

Under the system model, when an update is successfully generated, the local processor pulls a fresh sample from the source. Therefore, the age of the newly delivered update at the start of the next slot is exactly $X = 1$, yielding an expected value of $\E{X} = 1$. 

In the Memoryless inspection policy, an update is triggered with probability $p$ in each slot, and is successfully generated with probability $\mu_{\mathrm{m}}$. Thus, the inter-update time $S_i$ follows a geometric distribution $S \sim \Geomv(p\mu_{\mathrm{m}})$. This yields
\begin{align}
    \E{S} &= \frac{1}{p\mu_{\mathrm{m}}} \\
    \E{S^2} &= \frac{2 - p\mu_{\mathrm{m}}}{(p\mu_{\mathrm{m}})^2}
\end{align}

Substituting these moments and $\E{X} = 1$ back into the average age equation \eqref{eq:mem_age_expectation}, we get
\begin{align}
    \Delta &= 1 + \frac{1}{2} \left( \frac{\frac{2 - p\mu_{\mathrm{m}}}{(p\mu_{\mathrm{m}})^2}}{\frac{1}{p\mu_{\mathrm{m}}}} \right) \nonumber \\
    &= 1 + \frac{1}{2} \left( \frac{2 - p\mu_{\mathrm{m}}}{p\mu_{\mathrm{m}}} \right) \nonumber \\
    &= \frac{1}{p\mu_{\mathrm{m}}} + \frac{1}{2}
\end{align}

The average age at the monitor for this policy is given by 
\begin{align}
    \Delta_\mathrm{mem} & = \frac{1}{p\mu_\mathrm{m}} + \frac{1}{2}\eqnlabel{agememorylessapp}. 
\end{align}

\subsection{Proof of Lemma~\ref{lem:optimal_p}}
\label{app:lemma_optimal_p}
The total average cost is $\bar{C}(p) = \E{A} + w\E{Q}$. 
\begin{align}
    \bar{C}(p) = \left( \frac{1}{p\mu_{\mathrm{m}}}  + \frac{1}{2}\right) + w \frac{\lambda(1-\mu_{\mathrm{eff}}(p))}{\mu_{\mathrm{eff}}(p)-\lambda}
\end{align}
where $\mu_{\mathrm{eff}}(p) = \mu_{\mathrm{f}} - p\Delta\mu$. To find the minimum, we take the derivative with respect to $p$ and set it to zero, yielding:
\begin{align}
    \frac{1}{\mu_{\mathrm{m}}p^2} = \frac{w\lambda(1-\lambda)\Delta\mu}{(\mu_{\mathrm{f}} - p\Delta\mu - \lambda)^2}.
\end{align}
Solving for $p$ gives the optimal probability $p^*$.

\subsection{Derivation of the Myopic Threshold}
\label{app:myopic_derivation}

To find the optimal myopic action, we first evaluate the expected next state age and queue length under each action. For the age evolution, conditioned on the current age $A_n = a$, we have:
\begin{align}
     \E{A_{n+1} \mid A_n = a, U_n = 0} &= a + 1, \eqnlabel{eq:unequal0}\\
    \E{A_{n+1} \mid A_n = a, U_n = 1} &= 1 \cdot \mu_{\mathrm{m}} + (a+1)(1-\mu_{\mathrm{m}}) \nonumber\\
    &= a + 1 - a\mu_{\mathrm{m}} \eqnlabel{eq:unequal1}.
\end{align}

For the expected queue length, the evolution depends on whether the queue is currently empty or has jobs waiting in line:
\begin{align}
    \E{Q_{n+1} \mid Q_n = q, U_n = u} &= q + \lambda - \mu(u) \quad \text{for } q > 0, \eqnlabel{eq: qgreater0}\\
    \E{Q_{n+1} \mid Q_n = 0, U_n = u} &= \lambda(1-\mu(u))  \eqnlabel{eq: qequal0}.
\end{align}

The Myopic scheduler chooses to generate an update ($u=1$) if the expected next-step cost of updating is strictly less than the expected next-step cost of serving the queue:
\begin{align}
    \E{C_{n+1} \mid u=1} < \E{C_{n+1} \mid u=0}. \eqnlabel{eq:myopic_condition_app}
\end{align}

We evaluate \eqnref{eq:myopic_condition_app} for two cases based on the queue backlog.

Case 1 ($q > 0$): Substituting \eqnref{eq:unequal0}, \eqnref{eq:unequal1}, and \eqnref{eq: qgreater0} into \eqnref{eq:myopic_condition_app}, we get:
\begin{align}
    &(a + 1 - a\mu_{\mathrm{m}}) + w(q + \lambda - \mu_{\mathrm{s}}) \nonumber \\
    &\quad < (a + 1) + w(q + \lambda - \mu_{\mathrm{f}}).
\end{align}
Simplifying this expression yields the age threshold for a backlogged queue:
\begin{align}
    a &> \frac{w(\mu_{\mathrm{f}} - \mu_{\mathrm{s}})}{\mu_{\mathrm{m}}}.
\end{align}

Case 2 ($q = 0$): Similarly, substituting the expected values for an empty queue \eqnref{eq: qequal0} into \eqnref{eq:myopic_condition_app} yields:
\begin{align}
    &(a + 1 - a\mu_{\mathrm{m}}) + w(\lambda - \lambda\mu_{\mathrm{s}}) \nonumber \\
    &\quad< (a + 1) + w(\lambda - \lambda\mu_{\mathrm{f}}).
\end{align}
Simplifying this expression provides the threshold when the queue is empty:
\begin{align}
    a &> \frac{w\lambda(\mu_{\mathrm{f}} - \mu_{\mathrm{s}})}{\mu_{\mathrm{m}}}.
\end{align}

\subsection{Stability Analysis of the Myopic Policy}
\label{App:myp_stability}

We establish the stability of the Myopic policy by applying Corollary~1 from Foss et~al.~\cite{foss2013stability}. We model the system as a Markov-modulated process $Z_n = (A_n, Q_n)$, where the age process $A_n$ evolves autonomously but modulates the evolution of the queue $Q_n$. To account for the boundary conditions at $Q_n = 0$, we introduce a \textit{fictitious system} $(\tilde{A}_n, \tilde{Q}_n)$ where the scheduler applies a constant age threshold $K = \lfloor \frac{w(\mu_{\mathrm{f}} - \mu_{\mathrm{s}})}{\mu_{\mathrm{m}}} \rfloor$ regardless of the queue state. 

Following the framework in Foss et al.~\cite{foss2013stability}, we verify the essential conditions for the stability of the fictitious system:
\begin{itemize}
    \item Condition A (Modulator Ergodicity): The age process $\{\tilde{A}_n\}$ must be ergodic. In the fictitious system, the transition probabilities for the age are
    \begin{align}
        &P(\tilde{A}_{n+1} = a' \mid \tilde{A}_n = a) \nonumber \\
        &\quad = 
        \begin{cases} 
            1 & a \leq K, a' = a+1, \\
            \mu_{\mathrm{m}} & a > K, a' = 1, \\
            1-\mu_{\mathrm{m}} & a > K, a' = a+1. 
        \end{cases}
\end{align}
The age process $\{\tilde{A}_n\}$ follows a renewal cycle that resets to state 1 with finite expected time $\E{\tau} = K + 1/\mu_\mathrm{m}$ with every successful update. Thus the age process is positive recurrent and has a unique steady-state distribution $\pi(a)$. 
    \item Condition B1 (Bounded Drift): There exists a constant $U$ such that 
    \begin{align}
    \sup_{a, q} \E{|\tilde{Q}_{n+1} - \tilde{Q}_n| \mid \tilde{A}_n = a, \tilde{Q}_n = q} \leq U.
    \end{align}
    In our discrete-time model, $Q_{n+1} = (Q_n - D_n + X_n)^+$. Since arrivals $X_n \in \{0,1\}$ and departures $D_n \in \{0,1\}$,  $|Q_{n+1} - Q_n| <= 1$. Thus, the condition is satisfied with $U=1$.
    \item Limiting Drift: We define the Lyapunov function $L_2(q) = q$ and the conditional expected increment for the fictitious system as
\begin{align}
D(a, q) = \E{L_2(\tilde{Q}_{n+1}) - L_2(\tilde{Q}_n) \mid \tilde{A}_n = a, \tilde{Q}_n = q}.
\end{align}
This condition requires the existence of a limiting drift function $f(a)$ dependent only on the modulator, such that
\begin{align}
\lim_{q \to \infty} D(a, q) = f(a).
\end{align}
For our discrete-time system, when $\tilde{Q}_n \geq 1$, departures occur with probability $\mu(u)$ and arrivals with probability $\lambda$. Under the Myopic policy, the service rate $\mu(u)$ is deterministic for a given age $a$. Thus, the increment is strictly independent of $q$ for all $q \geq 1$, yielding the exact limit
\begin{align}
f(a) = \lambda - \mu(u) = 
\begin{cases} 
\lambda - \mu_{\mathrm{f}} & a \leq K, \\
\lambda - \mu_{\mathrm{s}} & a > K.
\end{cases}
\end{align}
Since $D(a, q) = f(a)$ for all $q \geq 1$, the conditional expected increment converges trivially in the $L_1$ norm to the limiting drift function $f(a)$ as $q \to \infty$. 
    
\end{itemize}
By verifying the above conditions we establish that the fictitious joint process $(\tilde{A}_n, \tilde{Q}_n)$ is stable (positive recurrent) if the average drift over the stationary distribution $\pi(a)$ is negative:
\begin{align}
\bar{D} = \sum_{a=1}^{\infty} f(a) \pi(a) < 0
\end{align}

The stationary probabilities for the fictitious age process are $\pi(a) = \frac{\mu_{\mathrm{m}}}{K\mu_{\mathrm{m}} + 1}$ for $a \in \{1, \dots, K\}$. The probability of the system being in Update Mode ($a > K$) is $\pi_{\text{Update}} = \frac{1}{K\mu_{\mathrm{m}} + 1}$. The stability condition becomes
\begin{align}
\left(\frac{K\mu_{\mathrm{m}}}{K\mu_{\mathrm{m}} + 1}\right)(\lambda - \mu_{\mathrm{f}}) + \left(\frac{1}{K\mu_{\mathrm{m}} + 1}\right)(\lambda - \mu_{\mathrm{s}}) < 0.
\end{align}
Rearranging terms yields the stability bound
\begin{align}
\lambda < \frac{K\mu_{\mathrm{m}}\mu_{\mathrm{f}} + \mu_{\mathrm{s}}}{K\mu_{\mathrm{m}} + 1}.
\end{align}
The original Myopic policy uses an age threshold $K_0 < K$ when $Q_n = 0$. This lower threshold only increases the frequency of updates when the queue is already empty, which does not decrease the service rate of backlogged jobs. Consider both the original system $Q_n$ and the fictitious system $\tilde{Q}_n$ experiencing identical job arrivals. By design, the fictitious system provides a lower effective service rate. Thus its queue length acts as a strict upper bound for the original queue, ensuring $Q_n \leq \tilde{Q}_n$ for all time slots $n$. Since the fictitious queue is proven to be positive recurrent, the smaller original system is inherently bounded and therefore positive recurrent as well.

\subsection{Proving Assumptions of Sennott~\texorpdfstring{\cite{sennott1989average}}{}}
\label{app:sennott}

\begin{lemma}
    For every state $s = (q,a)$ and discount factor $\beta \in(0,1)$, the discounted value function $V_\beta(s)$ is finite. 
\end{lemma}
\begin{proof}
    By Sennott~\cite{sennott1989average}, it is sufficient to show that there exists at least one stationary policy $f$ that induces an irreducible, ergodic Markov chain yielding a finite expected average cost:
\begin{align}
    \sum _{s} \pi_{s} C(s,f(s)) < \infty,
\end{align}
where $\pi_s$ is the steady-state distribution of the Markov chain.

Consider the Memoryless inspection policy where the scheduler triggers an update with a fixed probability $p$ in every slot. We choose $p$ such that $0 < p < \frac{\mu_{\mathrm{f}} - \lambda}{\mu_{\mathrm{f}} - \mu_{\mathrm{s}}}$. 

Under this policy, the queue service rate is decoupled from the state and becomes a constant effective rate $\mu_{\mathrm{eff}} = p \mu_{\mathrm{s}} + (1-p) \mu_{\mathrm{f}}$. Our choice of $p$ guarantees that the arrival rate is strictly less than the service rate ($\lambda < \mu_{\mathrm{eff}}$). Therefore, the queue evolves as a stable discrete-time Geo/Geo/1 queue. This queue is an irreducible, positive recurrent Markov chain with a finite steady-state expected length $\E{Q} < \infty$.

Simultaneously, the age process $A_n$ resets to $1$ with a strictly positive probability $p\mu_{\mathrm{m}}$ in every slot, and increases by $1$ otherwise. This forms a geometric renewal process. The age process is also positive recurrent with a finite steady-state expected age $\E{A} = \frac{1}{p\mu_{\mathrm{m}}} + \frac{1}{2} < \infty$.

 The queue and age processes are positive recurrent and any state $(q, a)$ can be reached (e.g., the queue can empty and the age can reset in a finite number of steps with positive probability), thus the joint Markov chain is irreducible and ergodic. The expected steady-state cost is
\begin{align}
    \sum_{s} \pi_s C(s,u) &= \E{A} + w\E{Q} < \infty.
\end{align}
Since there exists a policy with a finite average cost, the discounted cost $V_{\beta}(s)$ is finite for all $s$ and $\beta$.
\end{proof}

\begin{lemma} [Monotonicity in Age] 
\label{lemma:monotonicityage}
    For fixed $q$, $V_{\beta, n}(q, a)$ is monotonically non-decreasing in $a$.
\end{lemma}

\begin{proof}
    We prove this lemma by induction. For $n=1$, we have 
    \begin{align}
        V_{\beta, 1}(q, a)=0
    \end{align}
    for all $q\in \mathbb{N}_{0}$ and $a\in \mathbb{N}$.
    Thus the lemma holds for $n=1$. 
    Now let us say the lemma holds for some $n=k$. Then,
\begin{align}
    V_{\beta, k}(q, a)\le V_{\beta, k}(q, a+1)
\end{align}
for all $q\in \mathbb{N}_{0}$ and $a\in \mathbb{N}$.
Now we need to show the claim holds for $n=k+1$. We have
\begin{align}
    &V_{\beta, k+1}(q, a+1)- V_{\beta, k+1}(q, a) \nonumber \\
    &\quad = (a+1 + wq) - (a + wq) \nonumber \\
    &\qquad + \min \Bigl( \beta \E{V_{\beta, k}(Q'_0, a+2)}, \nonumber \\
    &\qquad \qquad \beta \mu_{\mathrm{m}} \E{V_{\beta, k}(Q'_1, 1)} \nonumber \\
    &\qquad \qquad + \beta (1-\mu_{\mathrm{m}}) \E{V_{\beta, k}(Q'_1, a+2)} \Bigr) \nonumber \\
    &\qquad - \min \Bigl( \beta \E{V_{\beta, k}(Q'_0, a+1)}, \nonumber \\
    &\qquad \qquad \beta \mu_{\mathrm{m}} \E{V_{\beta, k}(Q'_1, 1)} \nonumber \\
    &\qquad \qquad + \beta (1-\mu_{\mathrm{m}}) \E{V_{\beta, k}(Q'_1, a+1)} \Bigr), \label{eq:valuefunctiondiff_age}
\end{align}
for all $q\in\mathbb{N}_0$ and $a\in \mathbb{N}$.

Let
\begin{align}
    f'&=\beta \E{V_{\beta, k}(Q'_0, a+2)}\\
    f&=\beta \E{V_{\beta, k}(Q'_0, a+1)}\\
    g'&=\beta \mu_{\mathrm{m}} \E{V_{\beta, k}(Q'_1, 1)} \nonumber\\
    &\qquad + \beta (1-\mu_{\mathrm{m}}) \E{V_{\beta, k}(Q'_1, a+2)}\\
    g&=\beta \mu_{\mathrm{m}} \E{V_{\beta, k}(Q'_1, 1)} \nonumber\\
    &\qquad + \beta (1-\mu_{\mathrm{m}}) \E{V_{\beta, k}(Q'_1, a+1)}.
\end{align}
We see that
\begin{align}
    f'-f &= \beta \E {V_{\beta, k}(Q'_0, a+2)-V_{\beta, k}(Q'_0, a+1)}
\end{align}
and
\begin{align}
    g'-g &= \beta (1-\mu_{\mathrm{m}}) \E {V_{\beta, k}(Q'_1, a+2)-V_{\beta, k}(Q'_1, a+1)}.
\end{align}
From the induction assumption we know $V_{\beta, k}(Q', a+2)-V_{\beta, k}(Q', a+1)\ge 0$. Thus, we have $f'\ge f$ and $g'\ge g$. Going back to \eqref{eq:valuefunctiondiff_age}, we have
\begin{align}
    &V_{\beta, k+1}(q, a+1)- V_{\beta, k+1}(q, a)\nonumber\\
    &\qquad=1 + \min (f',g')- \min(f,g)\geq 0 \label{eq:fact1_age}
\end{align}
for all $q\in \mathbb{N}_0$ and $a\in \mathbb{N}$. The inequality \eqref{eq:fact1_age} follows from the property that $\min(f',g') \ge \min(f,g)$ when $f'\ge f$ and $g'\ge g$.

Thus the claim holds for $n=k+1$, completing the proof of Lemma \ref{lemma:monotonicityage}.
\end{proof}

\begin{lemma}[Monotonicity in Queue Length] \label{Lemma3}
    For fixed $a$, $V_{\beta, n}(q, a)$ is monotonically non-decreasing in $q$.
\end{lemma}

\begin{proof}
    We prove this lemma by induction. For $n=1$, we have 
    $$V_{\beta, 1}(q, a)=0$$
    for all $q\in \mathbb{N}_{0}$ and $a\in \mathbb{N}$.
    Thus the lemma holds for $n=1$. 
    Now let us say the lemma holds for some $n=k$. Then,
    $$V_{\beta, k}(q, a)\le V_{\beta, k}(q+1, a)$$
    for all $q\in \mathbb{N}_{0}$ and $a\in \mathbb{N}$.
    Now we need to show the claim holds for $n=k+1$. We have
    \begin{align}
    &V_{\beta, k+1}(q+1, a)- V_{\beta, k+1}(q, a) \nonumber\\
    &\quad = (a + w(q+1)) - (a + wq) \nonumber\\
    &\qquad + \min \Bigl( \beta \E{V_{\beta, k}(Q'_{0, q+1}, a+1)}, \nonumber\\
    &\qquad \qquad \beta \mu_{\mathrm{m}} \E{V_{\beta, k}(Q'_{1, q+1}, 1)} \nonumber\\
    &\qquad \qquad + \beta (1-\mu_{\mathrm{m}}) \E{V_{\beta, k}(Q'_{1, q+1}, a+1)} \Bigr) \nonumber\\
    &\qquad - \min \Bigl( \beta \E{V_{\beta, k}(Q'_{0, q}, a+1)}, \nonumber\\
    &\qquad \qquad \beta \mu_{\mathrm{m}} \E{V_{\beta, k}(Q'_{1, q}, 1)} \nonumber\\
    &\qquad \qquad + \beta (1-\mu_{\mathrm{m}}) \E{V_{\beta, k}(Q'_{1, q}, a+1)} \Bigr), \label{eq:valuefunctiondiff_queue}
\end{align}
    where $Q'_{u, q}$ denotes the next queue state given current queue $q$ and action $u$.
    
    Let
    \begin{align}
        f'&=\beta \E{V_{\beta, k}(Q'_{0, q+1}, a+1)}\\
        f&=\beta \E{V_{\beta, k}(Q'_{0, q}, a+1)}\\
        g'&=\beta \mu_{\mathrm{m}} \E{V_{\beta, k}(Q'_{1, q+1}, 1)} \nonumber\\
        &\qquad+ \beta (1-\mu_{\mathrm{m}}) \E{V_{\beta, k}(Q'_{1, q+1}, a+1)}\\
        g&=\beta \mu_{\mathrm{m}} \E{V_{\beta, k}(Q'_{1, q}, 1)} \nonumber\\
        &\qquad+ \beta (1-\mu_{\mathrm{m}}) \E{V_{\beta, k}(Q'_{1, q}, a+1)}.
    \end{align}
    
   Let $S \in \{0,1\}$ be the binary indicator of a potential service completion under action $u$, and $X \in \{0,1\}$ be the arrival indicator. The queue length in the next slot from state $q$ is:
    $$Q'_{u, q} = q - S \I{q > 0} + X$$
    Evaluating the difference between starting with $q+1$ jobs versus $q$ jobs under the exact same arrival and service realizations yields
    $$Q'_{u, q+1} - Q'_{u, q} = 1 - S (\I{q+1 > 0} - \I{q > 0})$$
    If the initial queue is empty ($q=0$), the difference is $1 - S \ge 0$, since $S \le 1$. If the initial queue is strictly positive ($q > 0$), the difference is $1 - S(0) = 1 \ge 0$. Therefore, for all $q \in \mathbb{N}_0$, we have
    $$Q'_{u, q+1} \ge Q'_{u, q}$$
    
    From the induction assumption, $V_{\beta, k}(q, a)$ is non-decreasing in $q$. Since the next state queue length is strictly non-decreasing with respect to the initial queue length, we have
    $$\E{V_{\beta, k}(Q'_{0, q+1}, a+1)} \ge \E{V_{\beta, k}(Q'_{0, q}, a+1)},$$
    which implies $f' \ge f$. Similarly, we can show $g' \ge g$.
    
    Going back to \eqref{eq:valuefunctiondiff_queue}, we have
    \begin{align}
        &V_{\beta, k+1}(q+1, a)- V_{\beta, k+1}(q, a)\nonumber\\
        &\qquad=w + \min (f',g')- \min(f,g)\geq w > 0 \label{eq:fact1_queue}
    \end{align}
    for all $q\in \mathbb{N}_0$ and $a\in \mathbb{N}$.
    
    Thus the claim holds for $n=k+1$, completing the proof of Lemma \ref{Lemma3}.
\end{proof}

Let $h_\beta (s)= V_\beta (s)- V_{\beta}(s_{\mathrm{ref}})$, where the reference state is defined as $s_{\mathrm{ref}} = (0,1)$.
\begin{lemma}\label{Lemma 4}
       There exists a non-negative $N$ such that $-N\le h_\beta (s) $ for all  $s\in \mathcal{S}$ and $\beta$. 
\end{lemma}
\begin{proof}
    For each $\beta$, we have 
\begin{align}
    V_{\beta, n}(q,a)\uparrow V_\beta(q,a) \eqnlabel{VnapproachV},
\end{align}
for every $q\in \mathbb{N}_0$ and $a\in \mathbb{N}$. 
From Lemma~\ref{lemma:monotonicityage} and \eqnref{VnapproachV}, we have
\begin{align}
    V_\beta (q,a)- V_\beta (q,1)\geq 0 \eqnlabel{ineq1}
\end{align} 
and from Lemma~\ref{Lemma3} and \eqnref{VnapproachV},  we have
\begin{align}
    V_\beta(q,1)- V_{\beta}(0,1)\geq 0. \eqnlabel{ineq2}
\end{align}
Adding \eqnref{ineq1} and \eqnref{ineq2}, we get
\begin{align}
    V_\beta(q,a)- V_\beta(0,1)\geq 0. \eqnlabel{ass2proof}
\end{align}
From \eqnref{ass2proof}, we see that the Lemma holds for $N=0$. 
\end{proof}

\begin{lemma} \label{Assumption 3}
     There exists a non-negative $M(s)$, such that $h_\beta (s)\le M(s)$ for every $s$ and $\beta$. For every $s$ there exists an action $u(s)$ such that $\sum _{s'} P_{ss'}(u(s)) M(s')<\infty$. In addition, $\sum_{s'} P_{ss'}(u) M(s')<\infty$ for all $s\in \mathcal{S}$ and $u\in \mathcal{U}$.
\end{lemma}
\begin{proof}
    The proof is same as that of Lemma 1~\cite{sennott1989average}. 
\end{proof}

\end{document}